\documentclass[copyright]{eptcs}
\providecommand{\event}{J.~Dassow, G.~Pighizzini, B.~Truthe (Eds.): 11th International Workshop\\
on Descriptional Complexity of Formal Systems (DCFS 2009)} 
\providecommand{\volume}{3}
\providecommand{\anno}{2009}
\providecommand{\firstpage}{107} 
\providecommand{\eid}{10}
\usepackage{breakurl}        

\input{dcfs.tex}

\usepackage{graphicx} 
\usepackage{mathrsfs}
\usepackage{amsmath}
\usepackage{amssymb}

\begin{document}

\title{On the Number of Membranes in Unary P~Systems}
\def\titlerunning{On the Number of Membranes in Unary P~Systems}
\def\authorrunning{R.~Freund, A.~Klein, M.~Kutrib}
 
\author{Rudolf Freund\,\thanks{corresponding author}
\institute{Fakult\"at f\"ur Informatik --
  Technische Universit\"at Wien\\
  Favoritenstra\ss{}e~9--11 -- A-1040 Wien -- Austria}
\email{rudi@emcc.at}
\and
Andreas Klein
\institute{Department of Pure Mathematics and Computer Algebra --
   Ghent University\\
   Krijgslaan 281-S22 -- B-9000 Ghent -- Belgium}
\email{klein@cage.ugent.be}
\and
Martin Kutrib
\institute{Institut f\"ur Informatik --
  Universit\"at Giessen\\
  Arndtstra\ss{}e~2 -- D-35392 Giessen -- Germany}
\email{kutrib@informatik.uni-giessen.de}
}

\maketitle

\begin{abstract}
We consider P systems with a linear membrane structure working on objects
over a unary alphabet using sets of rules resembling homomorphisms. Such a
restricted variant of P systems allows for a unique minimal representation
of the generated unary language and in that way for an effective solution 
of the equivalence problem. Moreover, we examine the descriptional 
complexity of unary P systems with respect to the number of membranes.
\end{abstract}

\section{Introduction}

One decade ago, inspired by biological systems to be found in nature,
Gheor\-ghe~P\u{a}un introduced membrane systems (e.\,g., see \cite{Paun1998}) as
a class of distributed parallel computing models, originally working on
multisets of objects. The main feature of membrane systems (soon called P
systems in the literature) are the hierarchically structured membranes
embedded in the outermost \textit{skin} membrane. Every membrane encloses a 
\textit{region} possibly containing other membranes as well as multisets of
specific objects that evolve according to \textit{evolution rules}
associated with the regions. A \textit{computation} is a sequence of
consecutive \textit{configurations} of the system evolving by applying the
evolution rules in parallel to the objects in every region in a maximal
manner. The vector of the multiplicities of objects present in the specified 
\textit{output membrane }in the final configuration of a halting computation
is considered to be the \textit{result} of such a computation; in that way,
a membrane system \textit{computes} a set of vectors of natural numbers.

Many variants of P systems are able to characterize the family of
recursively enumerable sets of vectors of natural numbers (the Parikh sets
associated with recursively enumerable languages). The first detailed
definitions and results can be found in \cite{Paun1998}, and an overview of
many variants is given in \cite{PaunRozenberg2002}. The first monograph on
membrane systems \cite{Paunbook2002} shows the great variety and potentials
of P systems; for the actual state of the art in membrane systems we refer
the interested reader to the P page \cite{Ppage}.

In \cite{Paunetal2002} it is shown that special variants of P systems with
the underlying membrane structure being only a linear tree can be
computationally complete when using adequate evolution rules. In \cite%
{Freund2001}, several quite simple variants of such P systems were
investigated that allowed for establishing an infinite hierarchy with
respect to the number of membranes, which usually cannot be the case for
computationally complete variants of P systems. These special variants of P
systems were inspired by the tree systems of morphisms (introduced in \cite%
{Dassowetal2001}) and the compound Lindenmayer systems (investigated in \cite%
{Dassow1986}). In this paper, we consider those variants as introduced in 
\cite{Freund2001}, yet especially over a unary alphabet, and investigate the
descriptional complexity of the unary languages generated by such systems.

\medskip

The rest of the paper is organized as follows: In the second section we give
some preliminary definitions of notions needed in this paper and then define
the special variant of self-reproducing P systems to be investigated in the
succeeding sections; moreover, we also construct a suitable example and
prove some first results. In the third section, we show how to construct a
unique minimal representation of the unary language generated by a unary P
system, which also allows us to show the decidability of the equivalence
problem for unary P systems in the succeeding section. The descriptional
complexity of unary P systems with respect to the number of membranes is
investigated in the fifth section.

\section{A specific variant of P systems}

In this section, after having specified some notions, we introduce the main
variant of P systems investigated in this paper, i.\,e., P systems with a
linear membrane structure and rules mapping each object to powers of itself.
Moreover, we give an illustrative example and prove two useful lemmas valid
for this specific variant of P systems.

\subsection{Preliminaries}

For the basic notions of formal language theory we refer the reader to
monographs such as \cite{handbook1997}. We just mention some specific 
notions used throughout the paper: The empty string is denoted 
by~$\lambda $, and for the length of a string $w$ we write $|w|$. The 
set of mappings from a set $M$ to a set $N$ is denoted by $N^{M}$. We 
denote the set of positive integers by ${\mathbb{N}}_{+}$ and the set of 
non-negative integers by ${\mathbb{N}}$; ${\mathbb{N}}^{n}$ denotes the set 
of vectors $\left( x_{1},\dots ,x_{n}\right) $ with $x_{i}\in {\mathbb{N}}$, $1\leq i\leq n$, and for $\left( x_{1},\dots ,x_{n}\right) \in 
{\mathbb{N}}^{n}$, $\left( x_{1},\dots ,x_{n}\right) ^{T}$ denotes the 
corresponding transposed vector. Moreover, we write ${\mathbb{P}}$ for the 
set of prime numbers, and throughout the paper, $p_{i}$ denotes the $i$-th 
prime number, i.\,e., $p_{1}=2$, $p_{2}=3$, $p_{3}=5$, \dots

\subsection{Self-reproducing P systems}

Now let us consider the following very special variant of P systems, where

\begin{itemize}
\item the underlying membrane structure is a linear tree $[_{0}[_{1}\dots
\lbrack _{n}\ ]_{n}\dots ]_{1}]_{0}$;

\item the only axiom is in the innermost membrane;

\item except for the skin membrane (labeled by $0$), each region $i\geq 1$
contains the membrane dissolving operation $\delta $ as well as for each
object $a$ exactly one (noncooperative) rule of the special form $%
a\rightarrow a^{m(i,a)}$ for some $m(i,a)>0$;

\item there are no rules in the skin membrane;

\item the result of a computation is the multiset finally appearing in the
skin membrane (which constitutes a halting computation).
\end{itemize}

There are no additional features used when applying the (noncooperative)
rules of the special form $a\rightarrow a^{m(i,a)}$ in a maximally parallel
manner; for example, there are no priority relations among the rules.

The movement of objects from the innermost membrane to the skin membrane
only happens by dissolving one membrane after the other one using the
membrane dissolving operation $\delta $ which may be chosen at any time for
a membrane containing at least one object.

\bigskip

These special P systems as described above were introduced in \cite%
{Freund2001} and had been inspired by the tree systems of morphisms
(introduced in \cite{Dassowetal2001}) and the compound Lindenmayer systems
(investigated in \cite{Dassow1986}). In the case of linear trees, these tree
systems with morphisms intuitively correspond with the self-reproducing
deterministic linear P systems as defined above. The main difference is the
(additional) skin membrane, which only serves for terminating the
computation and collecting its result.

\medskip

Due to this correspondence with tree systems with morphisms, we can easily
describe a P system of the special form as defined above in the following
way:%
\begin{equation*}
\Pi =\left( V,h_{1},\dots ,h_{n},w\right)
\end{equation*}%
where $V$ is an alphabet, $n$ is the height of the linear tree describing
the membrane structure, $w$ is the axiom over $V,$ and the $h_{i}$ are the
homomorphisms defined by the rules of the form $a\rightarrow a^{m(i,a)}$ in
region $i$. As each object is mapped to powers of itself only, we call these
P systems \textit{self-reproducing}.

\medskip

As it is easy to see, the language generated by a self-reproducing P system 
\begin{equation*}
\Pi =\left( V,h_{1},\dots ,h_{n},w\right)
\end{equation*}%
can be written as 
\begin{equation*}
h_{1}^{\ast }(h_{2}^{\ast }\cdots \left( h_{n}^{\ast }(w)\right) \cdots ),
\end{equation*}%
in other words, 
\begin{equation*}
L_{\ast }(\Pi )=\{h_{1}^{m_{1}}(h_{2}^{m_{2}}\cdots \left(
h_{n}^{m_{n}}(w)\right) \cdots )\mid m_{i}\geq 0,1\leq i\leq n\}.
\end{equation*}

We may also demand that every homomorphism has to be applied at least once
before the corresponding membrane may be dissolved; in that case, the
corresponding language generated by\linebreak
\hbox{$\Pi =\left( V,h_{1},\dots,h_{n},w\right) $} can be written as 
\begin{equation*}
h_{1}^{+}(h_{2}^{+}\cdots \left( h_{n}^{+}(w)\right) \cdots ),
\end{equation*}%
in other words, 
\begin{equation*}
L_{+}(\Pi )=\{h_{1}^{m_{1}}(h_{2}^{m_{2}}\cdots \left(
h_{n}^{m_{n}}(w)\right) \cdots )\mid m_{i}\geq 1,1\leq i\leq n\}.
\end{equation*}

\subsection{An illustrative example and first results}

\begin{example}
\label{Example1}For each $n\geq 1,$ consider the language 
\begin{equation*}
L\left( n\right) =\{a_{1}^{p_{1}^{m_{1}}}\dots a_{n}^{p_{n}^{m_{n}}}\mid
m_{i}\geq 0,1\leq i\leq n\}
\end{equation*}%
where $p_{i}$ denotes the $i$-th prime number.

$L\left( n\right) $ is generated by the P system%
\begin{equation*}
\Pi \left( n\right) =\left( V,h_{1},\dots ,h_{n},w\right)
\end{equation*}%
with $V=\{a_{1},\dots ,a_{n}\},$ $w=a_{1}\dots a_{n},$ $%
h_{i}(a_{i})=a_{i}^{p_{i}}$ and $h_{i}(a_{k})=a_{k}$ for $k\neq i,$ $1\leq
k\leq n,$ $1\leq i\leq n$.

Obviously, $a_{1}^{p_{1}^{m_{1}}}\dots a_{n}^{p_{n}^{m_{n}}}$ is obtained as 
$h_{1}^{m_{1}}(\cdots \left( h_{n}^{m_{n}}(a_{1}...a_{n}) \right) \cdots )$;
moreover, only elements from $L\left( n\right) $ can be generated by $\Pi
\left( n\right) $, hence, $L\left( n\right) =L_{\ast }(\Pi \left( n\right) )$%
. Finally, we observe that 
\begin{equation*}
L_{+}(\Pi \left( n\right) )=\{a_{1}^{p_{1}^{m_{1}}}\dots
a_{n}^{p_{n}^{m_{n}}}\mid m_{i}\geq 1,1\leq i\leq n\}.\tag*{$\diamond$}
\end{equation*}
\end{example}

As it was shown in \cite{Freund2001}, $L\left( n\right) $ cannot be
generated by a self-reproducing P system with less than $n$ homomorphisms
(or equivalently, less than $n+1$ membranes), i.\,e., these variants of P
systems establish an infinite hierarchy with respect to the number of
membranes.

\bigskip

We now establish two simple but nevertheless useful lemmas. The first
observation is that the ordering of the homomorphisms is irrelevant (which
directly follows from the commutativity of multiplication):

\begin{lemma}
\label{homoperm}Let $\Pi =\left( V,h_{1},\dots ,h_{n},w\right) $ be a
self-reproducing P system, $u$ be a permutation of $1,\dots ,n$, and $\Pi
^{\prime }=\left( V,h_{u\left( 1\right) },\dots ,h_{u\left( n\right)
},w\right) $. Then $L_{\ast }(\Pi )=L_{\ast }(\Pi ^{\prime })$ as well as $%
L_{+}(\Pi )=L_{+}(\Pi ^{\prime })$.
\end{lemma}

\begin{proof}
Let $a\in V$ and let $\left\vert w\right\vert _{a}$ denote the number of
objects $a$ in the axiom $w$ of~$\Pi $; then obviously 
\begin{equation*}
h_{u\left( 1\right) }^{m_{u\left( 1\right) }}(h_{u\left( 2\right)
}^{m_{u\left( 2\right) }}\cdots (h_{u\left( n\right) }^{m_{u\left( n\right)
}}(a^{\left\vert w\right\vert _{a}}))\cdots
)=h_{1}^{m_{1}}(h_{2}^{m_{2}}\cdots (h_{n}^{m_{n}}(a^{\left\vert
w\right\vert _{a}}))\cdots ),
\end{equation*}%
as, due to the commutativity of multiplication, $\left( a^{m}\right)
^{k}=a^{mk}=a^{km}=$ $\left( a^{k}\right) ^{m}$ for all $k,m\geq 1$.
\end{proof}

Using Lemma~\ref{homoperm}, next we show that we may move the first
application of every homomorphism into the axiom, i.\,e., for every
self-reproducing P system $\Pi $ there is an equivalent self-reproducing~P
system~$\Pi ^{\prime }$ such that $L_{+}(\Pi )=L_{\ast }(\Pi ^{\prime })$.

\begin{lemma}
\label{lem:plustostar}Let $\Pi $ be a self-reproducing P system. Then we can
effectively construct a self-reproducing~P system $\Pi ^{\prime }$ such that 
$L_{+}(\Pi )=L_{\ast }(\Pi ^{\prime })$.
\end{lemma}

\begin{proof}
For $\Pi =\left( V,h_{1},\dots ,h_{n},w\right) $ we obtain
\begin{eqnarray*}
L_{+}(\Pi ) &=&\{h_{1}^{m_{1}}(h_{2}^{m_{2}}\cdots (h_{n}^{m_{n}}(w))\cdots
)\mid m_{i}\geq 1,1\leq i\leq n\} \\
&=&\{h_{1}^{m_{1}-1}(h_{2}^{m_{2}-1}\cdots (h_{n}^{m_{n}-1}(h_{1}(\cdots
(h_{n}(w))\cdots )))\cdots )  \mid m_{i}\geq 1,1\leq i\leq n\} \\
&=&\{h_{1}^{m_{1}^{\prime }}(h_{2}^{m_{2}^{\prime }}\cdots \left(
h_{n}^{m_{n}^{\prime }}(w^{\prime })\right) \cdots )\mid m_{i}^{\prime }\geq
0,1\leq i\leq n\} \\
&=&L_{\ast }(\Pi ^{\prime }),
\end{eqnarray*}%
where $\Pi ^{\prime }=\left( V,h_{1},\dots ,h_{n},w^{\prime }\right) $ with $%
w^{\prime }=h_{1}(\cdots (h_{n}(w))\cdots )$.
\end{proof}

The following example shows that the converse of Lemma~\ref{lem:plustostar}
does not hold:

\begin{example}
Consider self-reproducing P systems over a one letter alphabet, i.\,e., 
\textit{unary P systems} of the form $\Pi =\left( \left\{ a\right\}
,h_{1},\dots ,h_{n},w\right) $ as considered in more detail in the 
succeeding section; without loss of generality, let us assume that 
none of the homomorphisms $h_{i}$ with $1\leq i\leq n$ equals the
identity $h_{i}(a)=a$, because it would have no effect on the results.
Clearly, the axiom $w$ belongs to the language $L_{\ast }(\Pi )$. Moreover,
there cannot be any other string in $L_{\ast }(\Pi )$ that is shorter than
the axiom. So, choosing $w=a^{p_{i}}$ and a sole homomorphism~$h_{1}$ such
that $h_{1}(a)=a^{p_{j}}$, where $p_{i}$ and $p_{j}$ are different prime
numbers, we obtain a system $\Pi =\left( \left\{ a\right\}
,h_{1},a^{p_{i}}\right) $ whose language $L_{\ast }(\Pi )$ is $%
\{a^{p_{i}p_{j}^{m}}\mid m\geq 0\}$. On the other hand, it is easy to see
that for any other unary system $\Pi ^{\prime }$ the string $a^{p_{i}}$
cannot belong to $L_{+}(\Pi ^{\prime })$ unless the axiom is $a$ and in $\Pi
^{\prime }$ there must be a homomorphism~$h^{\prime }$ with $h^{\prime
}\left( a\right) =a^{p_{i}}$. But in this case, $a^{p_{i}^{2}}$ belongs to $%
L_{+}(\Pi ^{\prime })$, too. This is a contradiction, since $p_{i}$ and $%
p_{j}$ are different prime numbers and, thus, $a^{p_{i}^{2}}$ cannot belong
\linebreak
to $L_{\ast }(\Pi )$.\hfill$\diamond$
\end{example}

\newpage

\section{Representation of unary P systems}

We now consider \textit{unary P systems}, i.\,e., systems of the form 
\begin{equation*}
\Pi =\left( \left\{ a\right\} ,h_{1},\dots ,h_{n},w\right)
\end{equation*}%
where $n\geq 1$, $w\in \left\{ a\right\} ^{\ast }$ is the axiom, and $%
h_{i}:\left\{ a\right\} \rightarrow \left\{ a\right\} ^{\ast }$, $1\leq
i\leq n$, are homomorphisms. Depending on whether or not each homomorphism
has to be applied at least once, as above we distinguish two languages
generated by $\Pi $: 
\begin{eqnarray*}
L_{\ast }(\Pi ) &=&\{h_{1}^{m_{1}}(h_{2}^{m_{2}}\cdots \left(
h_{n}^{m_{n}}(w)\right) \cdots )\mid m_{i}\geq 0,1\leq i\leq n\}, \\
L_{+}(\Pi ) &=&\{h_{1}^{m_{1}}(h_{2}^{m_{2}}\cdots \left(
h_{n}^{m_{n}}(w)\right) \cdots )\mid m_{i}\geq 1,1\leq i\leq n\}.
\end{eqnarray*}

Now we turn to a language representation which is advantageous for our
purposes. To this end, we consider the decomposition of positive integers
into prime factors. We denote by $F$ the set of mappings from $\mathbb{P}$
to $\mathbb{N}$ with finite support, i.\,e., only finitely many primes are
mapped to non-zero values. Then the decompositions are represented by the
mapping $pf:\mathbb{N}_{+}\rightarrow F$. For example, let $%
n=q_{1}^{e_{1}}\cdots q_{k}^{e_{k}}$, where $k\geq 1$, $q_{i}$ prime, $%
q_{i}\neq q_{j}$ for $i\neq j$, $e_{i}\geq 1$, $1\leq i,j\leq k$. Then $%
(pf(n))(q_{i})=e_{i}$, for $1\leq i\leq k$, and $(pf(n))(q)=0$, for 
\hbox{$q\notin\{q_{1},\dots ,q_{k}\}$}. For convenience, we write $pf_{n}$ for the mapping $%
pf(n)$. Clearly, it holds \hbox{$pf_{mn}=pf_{m}+pf_{n}$}. Next, let $gpf(m)$ be the
index of the greatest prime factor with non-zero exponent in the
decomposition of $m$, if $m>1$, and set $gpf(1)=1$.

Since $pf$ is injective, every unary language $L$ is uniquely represented by
the set $$P(L)=\{pf_{m}\mid a^{m}\in L\}.$$
Conversely, since $pf$ is
surjective, every subset $V\subseteq F$ uniquely represents a unary language 
\begin{equation*}
P^{-1}(V)=\{a^{m}\mid \mathrm{\ there\ exists\ }f\in V\mathrm{\ such\ that\ }%
m=p_{1}^{f(p_{1})}p_{2}^{f(p_{2})}p_{3}^{f(p_{3})}\cdots \}.
\end{equation*}%
(Recall that $p_{i}$ always denotes the $i$-th prime number.) Alternatively, 
$V$ may be a set of vectors of natural numbers, e.\,g., $(m_{1},\dots
,m_{n})^{T}\in V$ implies $a^{m}\in P^{-1}(V)$, where $m=p_{1}^{m_{1}}\cdots
p_{n}^{m_{n}}$. From this point of view, a vector of natural numbers may be
seen as a mapping from $F$ and vice versa. In the following, we will use both
notions in a synonymous way.

\medskip

Let $\Pi =\left( \{a\},w,h_{1},\dots ,h_{n}\right) $ be a system, where $%
h_{i}(a)=a^{c_{i}}$, for $c_{i}\geq 1$, $1\leq i\leq n$. We set\linebreak
\hbox{$k=\max\{gpf(|w|),gpf(c_{1}),\dots gpf(c_{n})\}$}
to be the largest prime factor
with non-zero exponent appearing in $|w|$ or one of the exponents $c_{i}$,
and define the matrix: 
\begin{equation*}
M_{\Pi }=\left( 
\begin{array}{cccc}
pf_{c_{1}}(p_{1}) & pf_{c_{2}}(p_{1}) & \cdots & pf_{c_{n}}(p_{1}) \\ 
pf_{c_{1}}(p_{2}) & pf_{c_{2}}(p_{2}) & \cdots & pf_{c_{n}}(p_{2}) \\ 
\vdots & \vdots & \ddots & \vdots \\ 
pf_{c_{1}}(p_{k}) & pf_{c_{2}}(p_{k}) & \cdots & pf_{c_{n}}(p_{k})%
\end{array}%
\right).
\end{equation*}

\medskip

\begin{lemma}
Let $\Pi =\left( \{a\},h_{1},\dots ,h_{n},w\right) $ be a unary P system.
Then 
\begin{equation*}
L_{+}(\Pi )=P^{-1}(M_{\Pi }(\mathbb{N}_{+}^{n})^{T}+b)\quad \mathrm{and}%
\quad L_{\ast }(\Pi )=P^{-1}(M_{\Pi }(\mathbb{N}^{n})^{T}+b),
\end{equation*}%
where $b=(pf_{|w|}(p_{1}),pf_{|w|}(p_{2}),\dots ,pf_{|w|}(p_{k}))^{T}$.
\end{lemma}

\begin{proof}
Let $k=\max \{gpf(|w|),gpf(c_{1}),\dots, gpf(c_{n})\}$ and $%
h_{i}(a)=a^{c_{i}} $, for \hbox{$c_{i}\geq 1$}, $1\leq i\leq n$, be defined as
above. There exist $(m_{1},\dots ,m_{n})\in \mathbb{N}_{+}^{n}$ and 
$(e_{1},\dots ,e_{k})^{T}\in M_{\Pi }(\mathbb{N}_{+}^{n})^{T}+b$ such that
the following is true: 
\begin{eqnarray*}
&&a^{m}\in L_{+}(\Pi ) \\
&\iff &a^{m}=h_{1}^{m_{1}}(h_{2}^{m_{2}}\cdots (h_{n}^{m_{n}}(w))\cdots ) \\
&\iff &m=|w|c_{n}^{m_{n}}c_{n-1}^{m_{n-1}}\cdots c_{1}^{m_{1}} \\
&\iff &pf_{m}(p_{i})=pf_{|w|}(p_{i})+{m_{n}}pf_{c_{n}}(p_{i})+\cdots +{m_{1}}%
pf_{c_{1}}(p_{i}),\quad 1\leq i\leq k \\
&\iff &pf_{m}(p_{i})=(M_{\Pi }(m_{1},m_{2},\dots ,m_{n})^{T}+b)[i],\quad
1\leq i\leq k \\
&\iff &pf_{m}(p_{i})=e_{i},\quad 1\leq i\leq k,\quad pf_{m}(q)=0,\quad
q\notin \{p_{1},\dots ,p_{k}\} \\
&\iff &m=p_{1}^{e_{1}}\cdots p_{k}^{e_{k}} \\
&\iff &a^{m}\in P^{-1}(M_{\Pi }(\mathbb{N}_{+}^{n})^{T}+b).
\end{eqnarray*}%
The assertion for $L_{\ast }(\Pi )$ follows in an analogous way.
\end{proof}

Due to the pumping lemma for regular languages, we immediately observe the
following result:

\begin{corollary}
The only context-free languages that can be generated by a unary P system
are of the form~$\{a^{m}\}$ for some $m\geq 1$.
\end{corollary}

\section{Equivalence and minimality of unary P systems}

We now may take advantage of the chosen language representation in order to
show the decidability of the equivalence problem for unary P systems.
Moreover, we obtain a unique minimal representation.

Consider the partial ordering 
\begin{equation*}
f\leq g\iff f(p)\leq g(p)\mathrm{\ for\ all\ }p\in \mathbb{P}
\end{equation*}%
defined on the set of mappings $\mathbb{N}^{\mathbb{P}}$. As usual, we call
an element $f$ minimal if $g\leq f$ implies the identity $g=f$. Let $\Pi
=\left( \left\{ a\right\} ,h_{1},\dots ,h_{n},w\right) $ be a unary P system
and $V=M_{\Pi }(\mathbb{N}^{n})^{T}+b$ be such that $L_{\ast }(\Pi
)=P^{-1}(V)$. In any other system $\Pi ^{\prime }=\left( \left\{ a\right\}
,h_{1}^{\prime },\dots ,h_{n}^{\prime },w^{\prime }\right) $ which generates
the same language, i.\,e., 
\hbox{$V^{\prime }=M_{\Pi ^{\prime }}(\mathbb{N}^{n^{\prime }})^{T}+b^{\prime }$} 
such that $P^{-1}(V^{\prime -1})=P^{-1}(V)$, 
the same prime factors as in $\Pi $ must occur. Hence, $M_{\Pi }$ and $%
M_{\Pi ^{\prime }}$ must have the same number of rows. Moreover, a
decomposition $b=b_{1}+b_{2}$ implies $b\geq b_{1}$. Therefore, $b$ is the
unique minimal element in $V$ and, thus, the unique minimal element in $%
V^{\prime }$, which implies $b=b^{\prime }$. For this reason, in order to
deal with equivalence and minimality, for what follows we may assume $b=0$
and consequently may minimize and compare sets of the form $M_{\Pi }(\mathbb{%
N}^{n})^{T}$.

An element $x\in M_{\Pi }(\mathbb{N}^{n})^{T}\setminus \{0_{k}\}$ -- where $%
0_{k}$ denotes the zero-vector from $(\mathbb{N}^{n})^{T}$ -- is called
irreducible, if there is no decomposition of the form $x=x_{1}+x_{2}$, where 
$x_{1},x_{2}\in M_{\Pi }(\mathbb{N}^{n})^{T}\setminus \{0_{k}\}$. The next
theorem leads to a unique minimal representation.

\begin{theorem}
\label{theo:minimal-representation} Let $\Pi $ be a unary P system with $n$
homomorphisms. If $M^{\prime }$ is the matrix whose columns are the $%
n^{\prime }$ irreducible elements of $M_{\Pi }(\mathbb{N}^{n})^{T}\setminus
\{0_{k}\}$, then 
\begin{enumerate}
\item[\rm (i)] all columns of $M^{\prime }$ are columns of $M_{\Pi }$ and 
\item[\rm (ii)] $M^{\prime }(\mathbb{N}^{n^{\prime }})^{T}=M_{\Pi }(\mathbb{N}^{n})^{T}$.
\end{enumerate}
\end{theorem}

\begin{proof}
Since an irreducible element of $M_{\Pi }(\mathbb{N}^{n})^{T}\setminus \{0\}$
cannot be decomposed into summands, it has a representation of the form $%
M_{\Pi }(0,\dots ,0,1,0,\dots ,0)^{T}$. Therefore, it is a column of $M_{\Pi
}$, and assertion (i) follows.

By (i), the inclusion $M^{\prime }(\mathbb{N}^{n^{\prime }})^{T}\subseteq
M_{\Pi }(\mathbb{N}^{n})^{T}$ follows immediately. Now assume that in
contrast to assertion (ii), the inclusion is a proper one. Then there is a
minimal element $x\in M_{\Pi }(\mathbb{N}^{n})^{T}\setminus M^{\prime }(%
\mathbb{N}^{n^{\prime }})^{T}$. Since $0_{k}\in M_{\Pi }(\mathbb{N}%
^{n})^{T}\cap M^{\prime }(\mathbb{N}^{n^{\prime }})^{T}$, we have $x\neq
0_{k}$. Moreover, $x$ is not irreducible, since otherwise it would be a
column of $M^{\prime }$ and, thus, would belong to $M^{\prime }(\mathbb{N}%
^{n^{\prime }})^{T}$. So, there is a decomposition \hbox{$x=x_{1}+x_{2}$}, where $%
x_{1},x_{2}\in M_{\Pi }(\mathbb{N}^{n})^{T}\setminus \{0\}$. This implies $%
x_{1}\leq x$ as well as $x_{2}\leq x$. Since $x$ is minimal in $M_{\Pi }(%
\mathbb{N}^{n})^{T}\setminus M^{\prime }(\mathbb{N}^{n^{\prime }})^{T}$,
both elements $x_{1}$ and $x_{2}$ belong to $M^{\prime }(\mathbb{N}%
^{n^{\prime }})^{T}$. Obviously, $M^{\prime }(\mathbb{N}^{n^{\prime }})^{T}$
is closed under addition. Therefore, $x$ belongs to $M^{\prime }(\mathbb{N}%
^{n^{\prime }})^{T}$, too. From this contradiction we infer assertion (ii).
\end{proof}

The preceding theorem shows that $M^{\prime }(\mathbb{N}^{n^{\prime
}})^{T}+b $ is the unique minimal representation of the language generated
by $\Pi $, i.\,e., $L_{\ast }(\Pi )=P^{-1}(M^{\prime }(\mathbb{N}^{n^{\prime
}})^{T}+b)$. It is easy to determine the irreducible columns out of the $n$
columns of $M_{\Pi }$. Hence, matrix $M^{\prime }$ can effectively be
constructed from $M_{\Pi }$. Conversely, given some $M^{\prime }$ with~$%
n^{\prime }$ columns and $b$, the unary P system $\Pi ^{\prime }$ whose
language is $P^{-1}(M^{\prime }(\mathbb{N}^{n^{\prime }})^{T}+b)$ can
effectively be constructed, too.

\begin{corollary}
There exists an effective algorithm which minimizes the number of
homomorphisms (membranes) of a unary P system.
\end{corollary}

\begin{corollary}
The equivalence of unary P systems is decidable.
\end{corollary}

\section{Descriptional complexity of unary P systems}

This section is devoted to descriptional complexity issues of unary P
systems with a different number of homomorphisms (membranes). The key
question is how succinct a language, given by some unary P system with $n$
homomorphisms, can be represented by some unary P system with at least $n$
homomorphisms.

In order to talk about the economy of descriptions we first have to define
what is meant by the size of a system. In general, we are interested to
measure the length of the string that defines a system. In particular, we
use more convenient size measures, such that there is a recursive upper
bound for the length of the defining string depending on the chosen size
measure. For example, the size of a finite automaton equals the product of
the number of its states and the number of its input symbols.

The \textit{size}~$\left\vert \Pi \right\vert $ of a unary P system $\Pi
=\left( \left\{ a\right\} ,h_{1},\dots ,h_{n},w\right) $ is defined to be 
\begin{equation*}
\left\vert \Pi \right\vert =\left\vert w\right\vert +\sum_{1\leq i\leq
n}\left\vert h_{i}(a)\right\vert .
\end{equation*}

We denote the family of unary P systems by $\mathcal{F}$. Clearly, the
considered measure implies a total, recursive function mapping a unary P
system to its size, such that $\mathcal{F}$ is recursively enumerable in
order of increasing size, and does not contain infinitely many members of
the same size.

Let $\mathcal{F}_{1}$ and $\mathcal{F}_{2}$ be two subfamilies of $\mathcal{F%
}$. A function $f:\mathbb{N}_{+}\rightarrow \mathbb{N}_{+}$, with $f(n)\geq
n $, is said to be an upper bound for the increase in size when changing
from a minimal description in~$\mathcal{F}_{1}$ to an equivalent minimal
description in $\mathcal{F}_{2}$, if 
\begin{equation*}
\min \{\left\vert \Pi \right\vert \mid \Pi \in \mathcal{F}_{2}\mathrm{\
generates\ }L\}\leq f(\min \{\left\vert \Pi \right\vert \mid \Pi \in 
\mathcal{F}_{1}\mathrm{\ generates\ }L\})
\end{equation*}%
for all languages $L$ generated by some system in $\mathcal{F}_{1}$ as well
as by some system in~$\mathcal{F}_{2}$.

The following theorem is an immediate consequence of Theorem~\ref%
{theo:minimal-representation}. Since any matrix of some representation
contains at least the columns of irreducible elements, any equivalent system
has at least the homomorphisms associated with these columns, respectively.

\begin{theorem}
Let $\Pi $ be a unary P system with $n$ homomorphisms where 
$$L_{\ast }(\Pi)=P^{-1}(M_{\Pi }(\mathbb{N}^{n})^{T}+b).$$
If the columns of $M_{\Pi }$ are
different irreducible elements of $M_{\Pi }(\mathbb{N}^{n})^{T}\setminus
\{0_{k}\}$, then every unary system $\Pi ^{\prime }$ with $L_{\ast }(\Pi
)=L_{\ast }(\Pi ^{\prime })$ is at most of size $\left\vert \Pi \right\vert $%
.
\end{theorem}

So, in case of $L_{\ast }(\Pi )$, the identity is an upper bound for the
increase of size when possibly changing from a description with $n$
homomorphisms to an equivalent description with $n^{\prime }<n$
homomorphisms. The situation is different in case of $L_{+}(\Pi )$:

\begin{theorem}
\label{theo:upper-bound} Let $\Pi $ be a unary P system with $n$
homomorphisms. If there is a unary P system $\Pi ^{\prime }$ with $n-1$
homomorphisms such that $L_{+}(\Pi )=L_{+}(\Pi ^{\prime })$, then $%
\left\vert \Pi ^{\prime }\right\vert \in O(\left\vert \Pi \right\vert ^{2})$.
\end{theorem}

\begin{proof*}
If a language $L=L_{+}(\Pi )$ is generated by some unary P system 
\begin{equation*}
\Pi =\left( \left\{ a\right\} ,h_{1},\dots ,h_{n},w\right)
\end{equation*}%
whose number of homomorphisms can be reduced, then there is at least one
column of $M_{\Pi }$ which is not an irreducible element of $M_{\Pi }(%
\mathbb{N}^{n})^{T}\setminus \{0_{k}\}$. Let $h_{n}$ be the associated
homomorphism, then 
\begin{equation*}
\Pi ^{\prime }=\left( \left\{ a\right\} ,h_{1},\dots ,h_{n-1},h_{n}\left(
w\right) \right)
\end{equation*}%
generates $L$, too, i.\,e., $L=L_{+}(\Pi ^{\prime })$. Therefore, we can
approximate the size of $\Pi ^{\prime }$ as follows:

Let $m=\max \{\left\vert h_{i}(a)\right\vert \mid 1\leq i\leq n\}$, then 
\begin{align*}
\left\vert \Pi ^{\prime }\right\vert 
&=\left\vert \Pi \right\vert -\left\vert
h_{n}(a)\right\vert -\left\vert w\right\vert +\left\vert h_{n}(w)\right\vert\\
&\leq \left\vert \Pi \right\vert -\left\vert h_{n}(a)\right\vert -\left\vert 
w\right\vert +m\left\vert w\right\vert \\
&\leq \left\vert \Pi \right\vert -1+(m-1)\left\vert w\right\vert\\
&\leq \left\vert \Pi \right\vert -1+(\left\vert \Pi \right\vert -1)*
\left\vert \Pi \right\vert \\
&=\left\vert \Pi \right\vert ^{2}-1\in O(\left\vert \Pi \right\vert ^{2}).\tag*{\qed}
\end{align*}
\end{proof*}

The next lemma shows that the upper bound of Theorem~\ref{theo:upper-bound}
can be reached.

\begin{lemma}
\label{lem_lower-bound} Let $\Pi =\left( \left\{ a\right\}
,h_{1},h_{2},w\right) $ where $h_{1}(a)=h_{2}(a)=a^{m}$ and $w=a^{m}$, for
some $m\geq 2$, then any unary P system $\Pi ^{\prime }$ with $L_{+}(\Pi
)=L_{+}(\Pi ^{\prime })$ has size $\Omega (\left\vert \Pi \right\vert ^{2})$.
\end{lemma}

\begin{proof}
Obviously, $h_{1}(a)=a^{m}$ is associated with a necessary column in~$M_{\Pi
}$. So, the sole equivalent unary P system with one homomorphism is 
\begin{equation*}
\Pi ^{\prime }=\left( \left\{ a\right\} ,h_{1},h_{2}\left( w\right) \right) .
\end{equation*}%
For the sizes we obtain $\left\vert \Pi \right\vert =3m$ and $\left\vert \Pi
^{\prime }\right\vert =m+m^{2}\in \Omega (\left\vert \Pi \right\vert ^{2})$.
\end{proof}

If we reduce the number of homomorphisms by more than one, i.\,e., we iterate
the construction $x$ times, then we do not obtain sizes of the form $%
O(\left\vert \Pi \right\vert ^{2^{x}})$. The reason is that we have to
multiply the sizes by at most a factor $\left\vert \Pi \right\vert $ in each
step.

\begin{corollary}
\label{cor:upper-bound} Let $\Pi $ be a unary P system with $n$
homomorphisms. If there is a unary P system $\Pi ^{\prime }$ with $n-x$
homomorphisms such that $L_{+}\left( \Pi \right) =L_{+}(\Pi ^{\prime })$,
then $\left\vert \Pi ^{\prime }\right\vert \in O(\left\vert \Pi \right\vert
^{x+1})$.
\end{corollary}

By an immediate generalization of the proof of Lemma~\ref{lem_lower-bound}
we obtain a matching lower bound (in the order of magnitude) in the worst
case.

\begin{lemma}
Let $\Pi =\left( \left\{ a\right\} ,h_{1},\dots ,h_{n},w\right) $ where $%
h_{1}(a)=\cdots =h_{n}(a)=a^{m}$ and $w=a^{m}$, for some $m\geq 2$, then any
unary P system $\Pi ^{\prime }$ with $n-x$ homomorphisms, $1\leq x\leq n-1$,
and $L_{+}(\Pi )=L_{+}(\Pi ^{\prime })$ has size~$\Omega (\left\vert \Pi
\right\vert ^{x+1})$.
\end{lemma}

Finally, if we fix the number $m$ and consider arbitrarily large $n$, then
the trade-off in the size may become exponential with respect to the base~$m$%
.

\begin{example}
Let $\Pi =\left( \left\{ a\right\} ,h_{1},\dots ,h_{n},w\right) $ where $%
h_{1}(a)=\cdots =h_{n}(a)=a^{m}$ and $w=a^{m}$, for some $m\geq 2$, then any
unary P system $\Pi ^{\prime }$ with one homomorphism and $L_{+}(\Pi
)=L_{+}(\Pi ^{\prime })$ has size $\Omega (\left\vert \Pi \right\vert ^{n})$.%
\hfill$\diamond$
\end{example}

\bigskip

The representation and the descriptional complexity of languages generated
by less restricted variants of P systems remains as a challenging task for
future research; for example, we may ask to which extent the results proved
above for unary P systems can also be formulated for self-reproducing P
systems.

\bibliographystyle{eptcs}
\bibliography{freund}

\end{document}